\documentclass[sigconf]{acmart}
\usepackage{booktabs}
\usepackage{amsmath}
\usepackage{amssymb}

\acmConference[WSDM '27]{The 20th ACM International Conference on Web Search and Data Mining}{February 2027}{Hong Kong}
\acmYear{2027}\copyrightyear{2027}

\begin{document}

\title{When Concealed Links Cannot Be Recovered: A Structural
Identifiability Bound and Evaluation Pitfalls in Offshore Leak Networks}

\author{Joseph Bingham}
\authornote{Corresponding author.}
\email{jbingham@campus.technion.ac.il}
\affiliation{%
  \institution{Technion -- Israel Institute of Technology}
  \city{Haifa}\country{Israel}}
\renewcommand{\shortauthors}{Bingham}

\begin{abstract}
Leaks such as the Panama and Paradise Papers expose large networks of offshore
entities, and they invite an obvious question for network learning. Can the
relations these structures are built to hide---above all, who beneficially owns
what---be recovered from the public part of the leak by link prediction? We argue
that the answer is mostly no, and that the analyses which suggest otherwise are
measuring the wrong thing. Our main result is a distribution-free identifiability
bound. For any recovery rule that respects graph isomorphism, and that covers
every topological link-prediction score together with every message-passing graph
neural network, a concealed endpoint left isolated in the observed graph is
interchangeable with its structural twins, so its hidden edge cannot be recovered
above chance. Isolation is only the sharpest case. In general the ceiling on
recovery is set by the size of a node's structural-indistinguishability class,
for which the Weisfeiler--Leman colour class is a computable stand-in, and node
degree is at best a loose proxy. On the full ICIJ Offshore Leaks graph (814K
entities and 84K labelled beneficial-owner edges) a classifier-free,
degree-controlled probe reproduces an exact $0.5$ floor for isolated owners,
who make up $20.6\%$ of all owners, and a trained graph neural network lands on
the same floor. Recovery climbs only as structural distinctiveness grows, and the
floor reappears in every one of the five leaks. Along the way we document five
evaluation traps. Each one makes a bound that cannot be beaten look beaten, and we
give a short rule that avoids them. The practical upshot is to redirect effort
from the hidden principal, which is close to unrecoverable, toward the machinery
of concealment, and to spell out why fusing external data helps far less than one
would hope.
\end{abstract}


\begin{teaserfigure}
  \includegraphics[width=\textwidth]{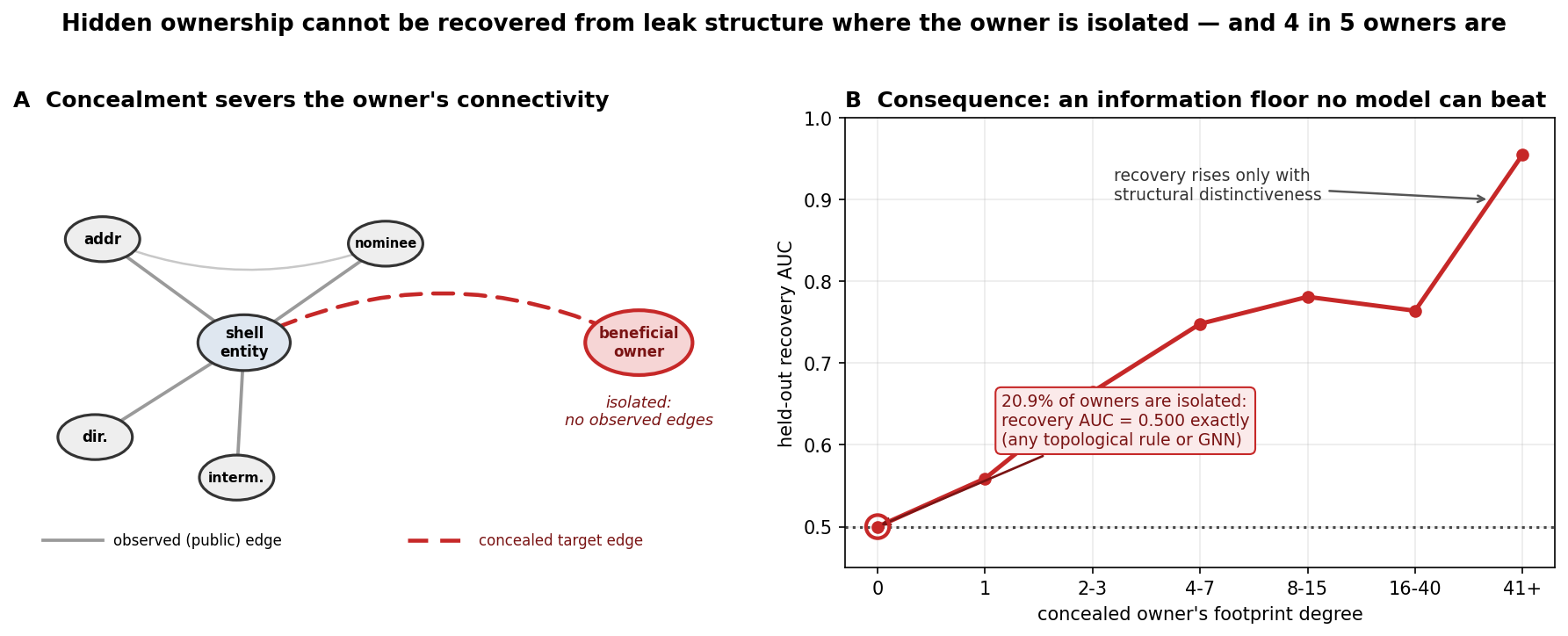}
  \caption{\textbf{The finding.} (A) Offshore structures hide ownership by cutting
  the principal off from the rest of the record. The beneficial owner appears as
  an isolated node whose sole tie to the shell entity is the concealed edge
  (dashed red), while administrative roles stay publicly linked. (B) The
  consequence, measured on all 84K beneficial-owner edges. Recovery sits at an
  exact $0.5$ floor for isolated owners, one in five of all owners, and climbs
  only with structural distinctiveness. The floor holds for every topological
  rule and every GNN.}
  \label{fig:abstract}
\end{teaserfigure}

\maketitle

\section{Introduction}
Offshore finance runs on anonymity. A beneficial owner---the real person who
controls and profits from a company---can sit behind layers of shell entities,
nominee directors, and intermediaries spread across jurisdictions that ask few
questions. When journalists obtain and publish a cache of these records, as with
the Offshore Leaks (2013), Panama Papers (2016), Bahamas Leaks (2016), Paradise
Papers (2017), and Pandora Papers (2021), the result is a rare public window into
that world. The International Consortium of Investigative Journalists (ICIJ) has
merged these caches into one graph of more than eight hundred thousand entities
and their officers, addresses, and intermediaries.

A graph like this is a natural target for network learning. Beneficial ownership
is the relation of interest, and it is exactly the relation the structures were
designed to obscure. If ownership left a recoverable trace in the surrounding
administrative scaffolding, then a link predictor could in principle reconstruct
it, and a static leak would become an engine for financial transparency. That
prospect is the motivation for a good deal of work on these networks, and it is
the prospect we examine.

Our conclusion is largely negative, and we think it is useful for being so.
Consider a shell company with a public record: a registered address, two
directors, an incorporating agent. Its true owner shows up once, as the endpoint
of a single ownership edge, and nowhere else. To recover that edge from
structure, some chain of observed links has to connect the owner back to the
company. When the owner is isolated, and one in five are, no such chain exists.
We prove that in this situation no structural method, however expressive, beats a
coin flip. The problem is not a weak model. The information is simply absent.

The target relation is well defined in the data. The ICIJ schema tags each
officer edge with a role, and the beneficial-ownership roles---ultimate
beneficial owner, beneficiary, owner-of---make up fewer than five percent of all
officer edges. The rest are administrative: shareholders, directors,
secretaries, none of them concealed. We treat the beneficial-owner edges as a
held-out target and the remaining graph as observed structure, and we ask when
the target can be recovered.

The paper makes four contributions. The first is a theorem. Using a
symmetry argument that needs no generative assumptions, we show that any
isomorphism-invariant recovery rule gives identical scores to structurally
identical nodes, so a concealed endpoint with no observed edges is recovered at
exactly chance, whatever the model. The second sharpens this beyond the isolated
case. Recovery is capped by the size of a node's structural-indistinguishability
class, which we make precise through the automorphism orbit and approximate in
practice by the Weisfeiler--Leman colour class. The third is empirical. On the
complete ICIJ graph we confirm the exact floor with bootstrap confidence
intervals, show a trained GNN and a SEAL-style high-order predictor obeying it,
and reproduce it independently across all five leaks. The fourth is
methodological. We catalogue five ways that ordinary evaluation choices
manufacture apparent recovery of a bound that cannot be beaten, and distil the
fix into one rule.

Two consequences follow. Recovering the hidden principal from leak structure is,
for most owners, out of reach, and the honest place to spend effort is the
machinery of concealment---the nominee services, mass-registration addresses, and
intermediary hubs, all of them densely connected by construction. And external
data helps far less than intuition suggests, because the owners one most wants to
identify are precisely those who appear nowhere else. We quantify both points.

\section{Background and Related Work}
\subsection{Preliminaries}
Write $G=(V,E,\tau)$ for a graph with vertices $V$, edges $E$, and a type map
$\tau:V\to\mathcal{T}$. A link-prediction score $s(u,v;G)$ ranks candidate pairs.
Topological scores include common neighbours $|N(u)\cap N(v)|$, Adamic--Adar
$\sum_{w\in N(u)\cap N(v)} 1/\log|N(w)|$, Katz, and rooted PageRank. The
Weisfeiler--Leman (WL) colour refinement updates each node's colour to a hash of
its own colour and the multiset of neighbour colours. Two nodes with different WL
colours can be told apart by some message-passing GNN, and WL colour classes
approximate automorphism orbits, so WL is the standard yardstick for what these
architectures can distinguish.

\subsection{Related work}
\textbf{Learning on offshore-leak networks.} Most prior work on the ICIJ graphs
ranks nodes by suspicion. Suspiciousness-Rank-Back-and-Forth, for instance,
scores nodes by their connectivity to known sanctioned parties~\cite{joaristi}.
These methods re-rank inside the given graph. They do not reconstruct a held-out
concealed relation, and they report no limit on identifiability.

\textbf{Link prediction and hidden-link recovery.} Classical topological
scores~\cite{adamicadar,katz} and GNN link predictors~\cite{kipf,sage} target
edges missing at random. The nearest relative of our setting is hidden-link
prediction in deliberately incomplete networks, such as recovering undisclosed
supply-chain relationships with relational GNNs~\cite{supplychain}. What sets our
problem apart is that concealment here works by cutting the principal's
connectivity, and that we prove an identifiability limit rather than report a
recovery rate.

\textbf{Link-inference and privacy attacks.} A separate line of work recovers
edges that were meant to stay hidden, treating recovery as a privacy attack on a
released model~\cite{stealing,linkteller}. We borrow its insistence on faithful
evaluation. Our result differs in kind, since it bounds what the observed graph
can reveal about the target relation, with no trained model in the loop.

\textbf{Cold-start and impossibility results.} An isolated concealed owner is a
cold-start node, one with no observed links, and cold-start link prediction is
hard for the same reason: structural signal is absent~\cite{coldstart}. We turn
that difficulty into an exact floor for a labelled real relation, and connect it
to the broader study of limits in network inference.

\textbf{Positive-unlabelled learning and entity resolution.} Leak extraction is
partial, so a missing beneficial-owner edge is not a confirmed negative, and the
labels are positive-unlabelled~\cite{pu}. Cross-source entity resolution is the
one channel that can, in principle, escape our bound, and its coverage sets the
ceiling we measure in Section~\ref{sec:external}.

\textbf{Graph expressivity.} Our distinctiveness result is stated through the WL
hierarchy~\cite{wl,gnnwl}. More expressive families, including subgraph
GNNs~\cite{seal,buddy}, refine the classes further but still cannot invent a path
where none exists, a point we return to in Sections~\ref{sec:bound}
and~\ref{sec:disc}.

\textbf{Evaluation pitfalls and scalable structure.} Recent work has surfaced
subtle leakage in link-prediction evaluation. SpotTarget shows that including
target edges at training or inference time distorts results, especially for
low-degree nodes~\cite{spottarget}, which sits naturally beside our trap
taxonomy. On the modelling side, scalable subgraph and structural-encoding methods
such as BUDDY~\cite{buddy}, S3GRL~\cite{s3grl}, and Bloom-signature edge
features~\cite{bloomsig} make high-order pair features cheap enough to run at our
scale, and they are the natural candidates for a full subgraph baseline. The use
of orbits for identification also appears in the causal-lifting view of link
prediction~\cite{causallift}, which shares our reliance on graph symmetry, though
for counterfactual rather than identifiability questions.

\section{Data and Problem Formulation}\label{sec:data}
\subsection{The ICIJ Offshore Leaks graph}
Table~\ref{tab:data} summarises the public ICIJ graph~\cite{icij}: $814{,}344$
entities, $771{,}315$ officers, $25{,}629$ intermediaries, and $402{,}246$
addresses, joined by roughly $3.3$M relationships. The dominant relations are
\texttt{officer\_of} ($1.72$M), \texttt{registered\_address} ($833$K), and
\texttt{intermediary\_of} ($599$K), alongside cross-source resolution edges and
the small \texttt{underlying} (nominee) relation.

\begin{table}[t]
\caption{The ICIJ Offshore Leaks graph. Beneficial-owner (BO) edges are the
target relation. Every other edge is observed.}
\label{tab:data}
\small
\begin{tabular}{@{}lr@{\hskip 18pt}lr@{}}
\toprule
\textbf{Node type} & \textbf{count} & \textbf{Edge / BO subtype} & \textbf{count}\\
\midrule
Entities      & 814{,}344 & \texttt{officer\_of}          & 1{,}720{,}357\\
Officers      & 771{,}315 & \texttt{registered\_address}  & 832{,}721\\
Intermediaries&  25{,}629 & \texttt{intermediary\_of}     & 598{,}546\\
Addresses     & 402{,}246 & resolution edges              & $\sim$166{,}000\\
\midrule
\multicolumn{2}{@{}l}{\emph{BO subtypes ($R$, total 84{,}172):}} & ultimate benef.\ owner & 25{,}883\\
 & & beneficiary of & 23{,}637\\
 & & owner of & 20{,}405\\
 & & beneficial owner(\,of) & 14{,}247\\
\bottomrule
\end{tabular}
\end{table}

\subsection{The target relation}
Each \texttt{officer\_of} edge carries a \texttt{link} sub-type. These separate
administrative roles---shareholder ($34\%$), director ($27\%$), and the
rest---from the beneficial-ownership roles in Table~\ref{tab:data}, which total
$R=84{,}172$ edges, under five percent of officer relations. The
beneficial-owner edges are our concealed target. Everything else forms the
observed graph.

\subsection{Provenance caveats}
Three facts about the data shape any honest evaluation. First, extraction is
partial. A great deal of beneficial ownership lives in unstructured documents and
never reaches the graph, which makes $R$ a lower bound and turns unlabeled pairs
into non-negatives. Second, ICIJ removes isolated nodes during publication, so
the degree distribution is truncated at the low end. Third, the five leaks are
separate collections, and we use them as independent replicates in
Section~\ref{sec:results}.

\subsection{Recovery rules and channels}
Let $R\cap E=\emptyset$, and let $s(o,e;G)$ score a candidate ownership pair. The
rule is structural, meaning isomorphism-invariant, when
$s(\pi(o),\pi(e);\pi\!\cdot\!G)=s(o,e;G)$ for every relabelling $\pi$. A rule can
draw on three channels of information: topology, internal attributes such as the
names already present in the leak, and external attributes joined from outside
sources. Our bound is about the topology channel. Section~\ref{sec:external}
measures the other two.

\section{The Recoverability Bound}\label{sec:bound}
\subsection{Setup}
For $v\in V$ let the footprint $d_{\mathrm{obs}}(v)$ be its degree in the observed
graph $G$, that is, the non-target edges left after $R$ is removed. We call an
officer \emph{isolated} when $d_{\mathrm{obs}}(v)=0$, meaning it has no observed
edge of any kind. Isolation is the sharpest case, and it holds for $20.6\%$ of
owners. It should not be confused with being \emph{disconnected} from a
particular entity, that is, lying in a different connected component of $G$, which
is a weaker condition that holds for $76.2\%$ of owner--entity pairs. The bound
below bites hardest on isolated owners and, through the orbit argument, on
structurally indistinguishable owners more generally. Let $\mathrm{Aut}(G)$ be the
group of type- and edge-preserving automorphisms, and for a target entity $e$ let
$\mathrm{Stab}(e)=\{\pi\in\mathrm{Aut}(G):\pi(e)=e\}$ be the subgroup fixing $e$.
Write $H\!\cdot\!o=\{\pi(o):\pi\in H\}$ for the orbit of $o$ under a subgroup $H$.
Nodes in a common orbit are structurally indistinguishable: nothing computed from
$(G,\tau)$ separates them.

\subsection{Main theorem}
\begin{theorem}[Indistinguishability floor]\label{thm:floor}
Fix a target entity $e$, let $H=\mathrm{Stab}(e)$, and let $C=H\!\cdot\!o$ be the
orbit of the true owner $o$ under $H$. For every structural rule $s$, the score
$s(\cdot,e;G)$ is constant on $C$. Two consequences follow.
\begin{enumerate}
\item[(i)] Ranking $o$ against a uniformly random member of $C$ gives expected
ROC-\textsc{auc} exactly $\tfrac12$.
\item[(ii)] The probability that $o$ is ranked strictly first among all
candidates is at most $1/|C|$.
\end{enumerate}
\end{theorem}
\begin{proof}
Take any $o',o''\in C$. By definition of the orbit there is $\pi\in H$ with
$\pi(o')=o''$, and because $\pi$ fixes $e$ we have $\pi(e)=e$. Isomorphism
invariance then gives $s(o',e;G)=s(\pi(o'),\pi(e);G)=s(o'',e;G)$, so $s(\cdot,e;G)$
takes a single value across $C$. A constant score orders the members of $C$
uniformly at random. For (i), the expected \textsc{auc} of one distinguished
member against a random other under a uniform order is $\tfrac12$. For (ii), the
true owner is equally likely to occupy any position among the $|C|$ tied
candidates, so it lands first with probability at most $1/|C|$.
\end{proof}

The argument uses only the symmetry of the realised graph. It assumes no
random-graph model, and it speaks to identifiability rather than to any one
estimator. Part (ii) states the general ceiling the coauthor-facing reader may
expect from an orbit argument: recovery is bounded by how many structural twins
the concealed endpoint has.

\begin{corollary}[Isolation floor]\label{cor:iso}
For a connected target $e$, every degree-$0$ officer of a common type lies in a
single class $C$, because the symmetric group on those isolated officers is a
subgroup of $\mathrm{Stab}(e)$: permuting edgeless vertices disturbs no edge and
fixes $e$. By Theorem~\ref{thm:floor}(i) an isolated true owner is tied with all
isolated decoys, and recovery is exactly $\tfrac12$.
\end{corollary}

The permutation in Corollary~\ref{cor:iso} treats isolated officers of one type
as interchangeable. Attributes tied to the schema, such as names or
jurisdictions, refine that type and shrink $\mathrm{Aut}(G)$, and a unique name
can pull an isolated node into a class of its own. That route is not a topological
escape but the internal-attribute channel of Section~\ref{sec:external}, worth
about seven percent. The theorem holds for the typed graph taken without those
attributes, and Section~\ref{sec:disc} treats their effect directly.

\begin{corollary}[Zero information]\label{cor:mi}
Let $T$ be the true owner of $e$, a random element of $C$ under any prior. Since
$\mathrm{Stab}(e)$ contains a subgroup acting transitively on $C$ and fixing $G$,
the observed graph is invariant to relabelling within $C$, and
$I(T;G\mid T\!\in\!C)=0$. No decision rule beats the prior. The impossibility is
about information, and it is independent of model capacity.
\end{corollary}

\subsection{Necessary, not sufficient}
A footprint is required but does not suffice. What matters is whether the
footprint \emph{reaches} $e$ through a short path. Footprint degree upper-bounds
reach, and reach upper-bounds recovery. Section~\ref{sec:results} shows a leak
(Pandora) whose owners carry ample footprint yet almost never reach their
entity, and whose recovery stays on the floor as a result.

\subsection{From orbits to a computable proxy}
Theorem~\ref{thm:floor}(ii) ties recovery to $|C|$, the size of the
indistinguishability class. The exact class is an automorphism orbit under
$\mathrm{Stab}(e)$, which is expensive to compute at this scale. The WL colour
class is a coarser, cheap surrogate: orbit-mates always share a WL colour, and
for almost all graphs the two coincide. Higher footprint tends to give a finer WL
colour and thus a smaller class, which is why degree serves as a weak proxy for
distinctiveness. Section~\ref{sec:results} tests the prediction that recovery
tracks WL-class size more tightly than it tracks degree.

\section{Experimental Setup}\label{sec:setup}
\textbf{Data and preprocessing.} We use the full public ICIJ CSV export. The
observed graph is built from every non-$R$ edge, and the $84{,}172$
beneficial-owner edges are held out as targets. No target edge ever appears in
the graph used to compute features.

\textbf{Classifier-free probe.} To test Theorem~\ref{thm:floor} without the
confound of model choice, we use a paired probe that holds the concealed endpoint
fixed. For each $(o,e)\in R$ we measure $o$'s structural reach---Adamic--Adar
over the observed graph, plus a two-hop term weighted by $0.5$, with
neighbourhoods capped at $300$ for tractability---to the true entity $e$ and to a
degree-matched decoy owner of $e$. The paired win-rate is the \textsc{auc}. We
stratify by footprint degree and report bootstrap $95\%$ intervals over $2000$
resamples. Because reach is zero for an isolated owner regardless of the partner,
the floor is the same whichever endpoint we vary (Figure~\ref{fig:modelindep}).

\textbf{GNN.} We train a SIGN/SGC-style network~\cite{sgc,sign}. Multi-hop
features $\mathrm{emb}=[X,\hat A X,\hat A^2 X]$ propagate over the observed graph,
and a gradient-boosted decoder scores pairs. The encoder is inductive and holds
no learnable per-node embeddings, so it cannot memorise targets, and an isolated
node collapses to a generic vector. Features are type one-hot and $\log$-degree,
with no names or jurisdictions, so recovery is purely structural. Targets are
split into disjoint $20$K train and $20$K test sets, all targets are removed from
the message-passing graph, and training negatives are uniform. The faithful
evaluation fixes $e$ and ranks the true owner against a degree-matched decoy drawn
from the beneficial-owner population, which closes both marginal leaks of
Section~\ref{sec:traps}.

\textbf{WL.} We run three rounds of randomised colour refinement over all
$2.0$M nodes of the observed graph, from which every target edge has already been
removed, so the partition cannot leak label information. Each colour maps to a
random value, neighbour sums propagate through
sparse matrix--vector products, and rounding to six decimals keeps identical
neighbour multisets merged. Results match across seeds to three decimals, so the
refinement reproduces the canonical partition and hash collisions do not affect
it.

\textbf{Reproducibility.} Everything runs on commodity CPU. Code to reproduce
each table and figure will be released.

\section{Results}\label{sec:results}
\subsection{The bound holds, hardened}
Table~\ref{tab:hardened} and Figure~\ref{fig:hardened} give the probe over all
$84{,}172$ edges. At footprint degree $0$ the \textsc{auc} is $0.500$ with a
zero-width confidence interval across $17{,}303$ edges, every comparison a tie,
climbing to $0.955$ at degree $41$ and above. Isolated owners, who sit provably at
the floor, are $20.6\%$ of the total. Recovery follows reachability stratum by
stratum, which points to reach as the driver rather than model capacity.

\begin{table}[t]
\caption{Recovery \textsc{auc} against the concealed owner's footprint degree,
all 84K edges, bootstrap $95\%$ CI. Degree $0$ is exactly the floor.}
\label{tab:hardened}
\small
\begin{tabular}{@{}lrcc@{}}
\toprule
\textbf{footprint deg.} & $n$ & \textbf{recovery \textsc{auc}} & \textbf{95\% CI}\\
\midrule
$0$      & 17{,}303 & \textbf{0.500} & $[0.500,0.500]$\\
$1$      & 32{,}614 & 0.559 & $[0.557,0.560]$\\
$2$--$3$ & 29{,}542 & 0.665 & $[0.662,0.667]$\\
$4$--$7$ &  3{,}417 & 0.748 & $[0.740,0.757]$\\
$8$--$15$&    881   & 0.781 & $[0.765,0.797]$\\
$16$--$40$&   239   & 0.764 & $[0.732,0.793]$\\
$\ge 41$ &    176   & 0.955 & $[0.932,0.974]$\\
\bottomrule
\end{tabular}
\end{table}

\begin{figure*}[t]
  \centering
  \includegraphics[width=\textwidth]{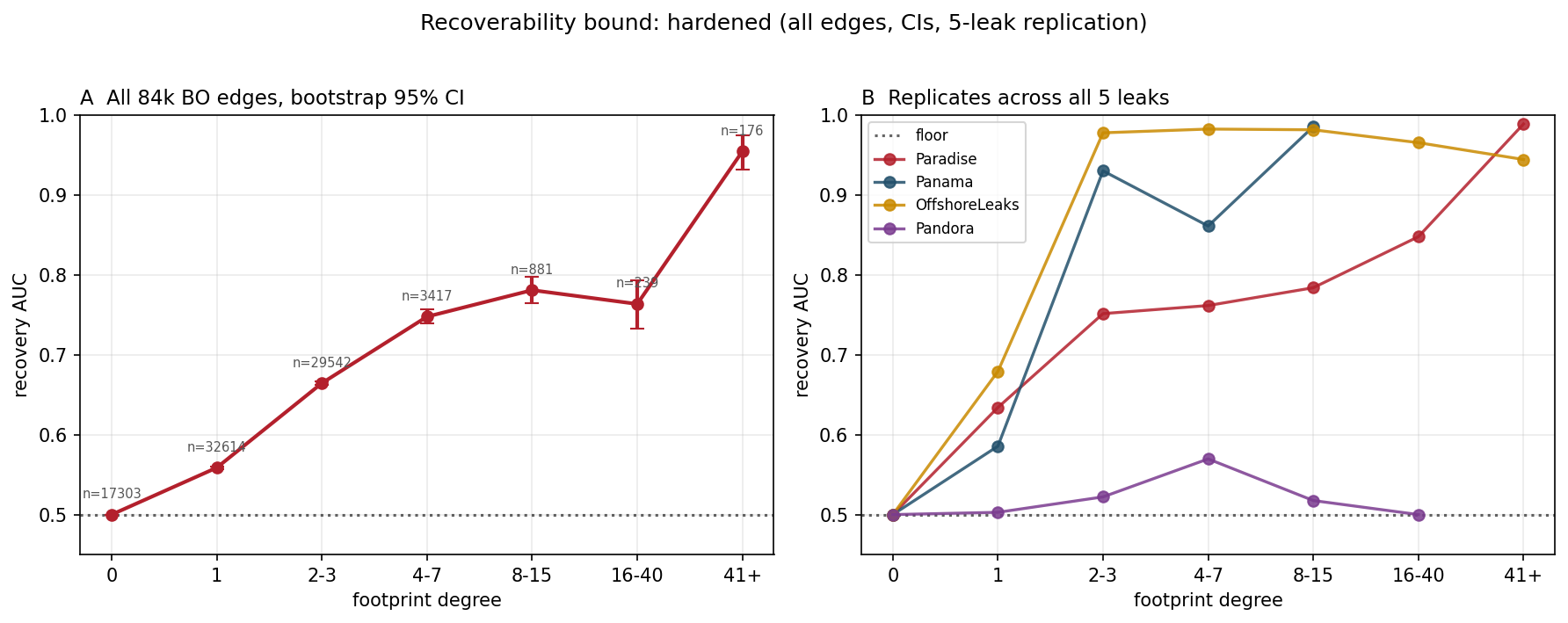}
  \caption{Recoverability against footprint degree. Left: all 84K edges with
  bootstrap $95\%$ CIs, and the degree-$0$ floor at exactly $0.500$. Right: the
  floor reappears in every leak that carries beneficial-owner labels, and Pandora
  stays on it because its owners have footprint but no reach.}
  \label{fig:hardened}
\end{figure*}

\subsection{Per-leak replication and a natural experiment}
The five leaks are separate collections (Table~\ref{tab:perleak}). The floor at
degree $0$ shows up in all four leaks that carry beneficial-owner labels. Bahamas
carries none, being a corporate registry that records directors and shareholders
rather than owners. The climb above the floor appears in three leaks but not in
Pandora.

Define the reachability fraction at radius $L$ as the share of pairs
$(o,e)\in R$ joined by an observed path of length at most $L$.
Table~\ref{tab:reach} reports it for the two largest leaks. Paradise owners reach
their entity more and more often as footprint grows, up to $0.99$ at radius three,
whereas Pandora owners, for all their footprint, stay at or below $0.14$ at every
degree. Pandora is thus a natural experiment separating footprint from reach, and
it is reach that governs recovery. Different providers appear to record ownership
in ways that leave the footprint pointing elsewhere. The typical Pandora pattern
is an owner attached to addresses and intermediaries that serve a \emph{different}
set of entities than the one owned, so the owner's neighbourhood and the target
entity's neighbourhood never meet within a short radius. For a practitioner this
is a useful diagnostic. When an owner's footprint and the target's scaffolding
share no near neighbours, the record is non-recoverable by construction, whatever
the model.

\begin{table}[t]
\caption{Per-leak replication. The degree-$0$ floor holds everywhere. The climb
holds except in Pandora, where footprint does not reach the target.}
\label{tab:perleak}
\small
\begin{tabular}{@{}lrcc@{}}
\toprule
\textbf{Leak} & \textbf{BO edges} & \textbf{deg-$0$ \textsc{auc}} & \textbf{top-bucket \textsc{auc}}\\
\midrule
Paradise       & 25{,}981 & 0.500 & 0.989 (\,$\ge$41)\\
Panama         & 15{,}245 & 0.500 & 0.986 (8--15)\\
Offshore Leaks &  5{,}823 & 0.500 & 0.944 (\,$\ge$41)\\
Pandora        & 37{,}123 & 0.500 & 0.500 (16--40)\\
Bahamas        &        0 & --- & --- (registry; no BO labels)\\
\bottomrule
\end{tabular}
\end{table}

\begin{table}[t]
\caption{Reachability fraction (share of owner--entity pairs joined by an observed
path of length $\le L$) by owner footprint degree. Pandora owners have footprint
but do not reach the entity they own.}
\label{tab:reach}
\small
\begin{tabular}{@{}lcccc@{}}
\toprule
& \multicolumn{2}{c}{\textbf{Paradise}} & \multicolumn{2}{c}{\textbf{Pandora}}\\
\cmidrule(lr){2-3}\cmidrule(lr){4-5}
footprint deg. & $L\!\le\!2$ & $L\!\le\!3$ & $L\!\le\!2$ & $L\!\le\!3$\\
\midrule
$1$       & 0.08 & 0.27 & 0.00 & 0.01\\
$2$--$3$  & 0.07 & 0.50 & 0.02 & 0.05\\
$4$--$7$  & 0.04 & 0.52 & 0.02 & 0.14\\
$8$--$15$ & 0.04 & 0.57 & 0.00 & 0.04\\
$16$--$40$& 0.12 & 0.70 & 0.00 & 0.00\\
$\ge41$   & 0.35 & 0.99 & --- & ---\\
\bottomrule
\end{tabular}
\end{table}

\subsection{A trained GNN obeys the floor}\label{sec:gnn}
Under the faithful evaluation---fix the entity, rank the true owner against a
degree-matched beneficial-owner decoy---the GNN scores exactly $0.500$ at degree
$0$ and rises only to the $0.67$--$0.87$ range where footprint exists
(Figure~\ref{fig:gnn}). The model is competent. It just cannot invent information
that is not present, as Corollary~\ref{cor:mi} requires. The floor holds for any
architecture, because every message-passing GNN aggregates over neighbours, and an
isolated node has none.

Two evaluation choices inflate these numbers, and we control for both. Varying the
entity rather than the owner lets the decoder read the entity's marginal
propensity to carry a hidden owner, which reads a spurious $0.765$ at degree $0$
(Trap~4). The mirror-image mistake is to draw the decoy owner from all officers
instead of from beneficial owners. That lets the decoder reward ``looks like an
owner'' and inflates degree-$1$ recovery from a true $0.670$ to $0.876$ (Trap~5,
Section~\ref{sec:traps}). Both leaks vanish under the same-population,
vary-the-owner pairing, and both leave the degree-$0$ floor at $0.500$, since an
isolated owner has a constant embedding no matter how the decoy is drawn.

\begin{figure}[t]
  \centering
  \includegraphics[width=\linewidth]{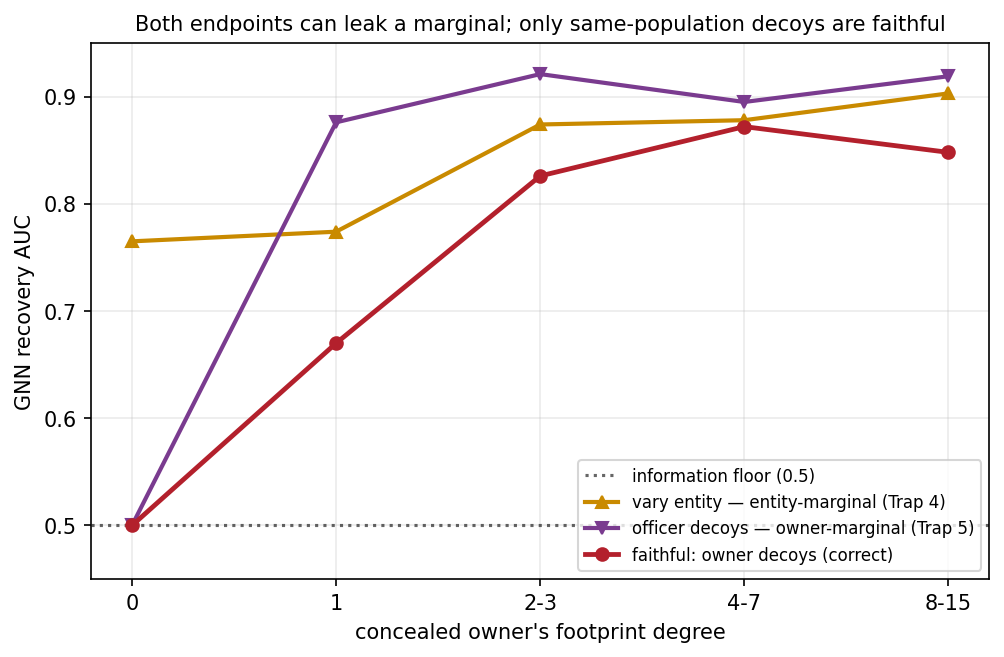}
  \caption{Either endpoint can leak a marginal. Varying the entity (Trap~4) or
  drawing the decoy owner from all officers (Trap~5) inflates GNN recovery. Only
  the same-population, vary-the-owner pairing is faithful, and all three curves
  meet at the $0.500$ floor.}
  \label{fig:gnn}
\end{figure}

\begin{figure}[t]
  \centering
  \includegraphics[width=\linewidth]{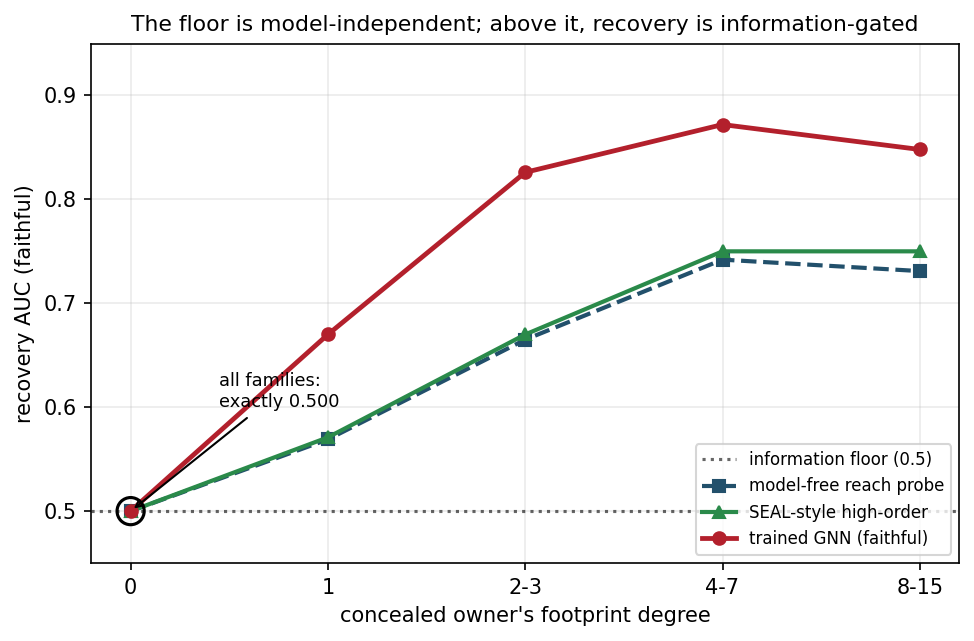}
  \caption{Three predictor families of growing expressivity---a model-free reach
  probe, a high-order SEAL-style predictor (Section~\ref{sec:seal}), and a trained
  GNN---under the faithful evaluation. All meet at $0.500$ at degree $0$. Above it
  the more expressive models recover a little more, and all stay well short of
  perfect, because recovery is bounded by reachable information.}
  \label{fig:modelindep}
\end{figure}

\subsection{Recovery tracks structural distinctiveness}
Binning owners by WL colour-class size, recovery falls as class size grows, from
$0.795$ for uniquely-coloured owners down to $0.519$ for classes of a thousand or
more. The sharper test holds degree fixed and varies class size within a band
(Figure~\ref{fig:wl}). Among degree-$2$--$3$ owners recovery runs from $0.795$
down to $0.60$ as the class grows, and among degree-$4$--$7$ owners from $0.81$
down to $0.50$. Class size is the better predictor of the two: its rank
correlation with recovery is $0.48$, against $0.39$ for degree, and $0.44$ against
$0.31$ once the isolated mass is set aside. This is what Theorem~\ref{thm:floor}(ii)
predicts, with WL classes standing in for orbits.

\textbf{Sensitivity.} The distinctiveness result is stable in depth and in
randomisation. Over one to five WL rounds and three seeds the outcome matches to
three decimals, since the random hashing only breaks ties. Colour counts converge
($7$K, $324$K, $699$K, $847$K, $871$K over the five rounds), and the signal grows
stronger with depth: the rank correlation rises from $0.39$ to $0.56$ across
rounds two through five while degree stays at $0.389$. As refinement approaches
orbits, the largest classes settle at exactly $0.500$, the floor of
Corollary~\ref{cor:iso}. A single round is too coarse to express distinctiveness,
and the effect emerges at two.

\textbf{WL-matched decoys.} Degree matching is a loose stand-in for orbit
matching. As a sharper test we draw the decoy owner from the true owner's WL
colour class, a structural near-twin, and repeat the paired probe. The floor is
untouched: at degree $0$ every isolated owner shares one colour, so the decoy is
another isolated owner and recovery is exactly $0.500$. Above the floor, matching
on colour makes the decoy harder to tell apart, and recovery comes in a little
below the degree-matched figure at every band ($0.542$ against $0.557$ at degree
$1$, $0.639$ against $0.662$ at $2$--$3$, $0.763$ against $0.791$ at $4$--$7$).
The availability of a twin is itself telling. All isolated owners have one,
whereas only a fifth of degree-$4$--$7$ owners do, the rest being structurally
unique. Distinctiveness and recoverability move together, exactly as
Theorem~\ref{thm:floor}(ii) says.

\begin{figure*}[t]
  \centering
  \includegraphics[width=\textwidth]{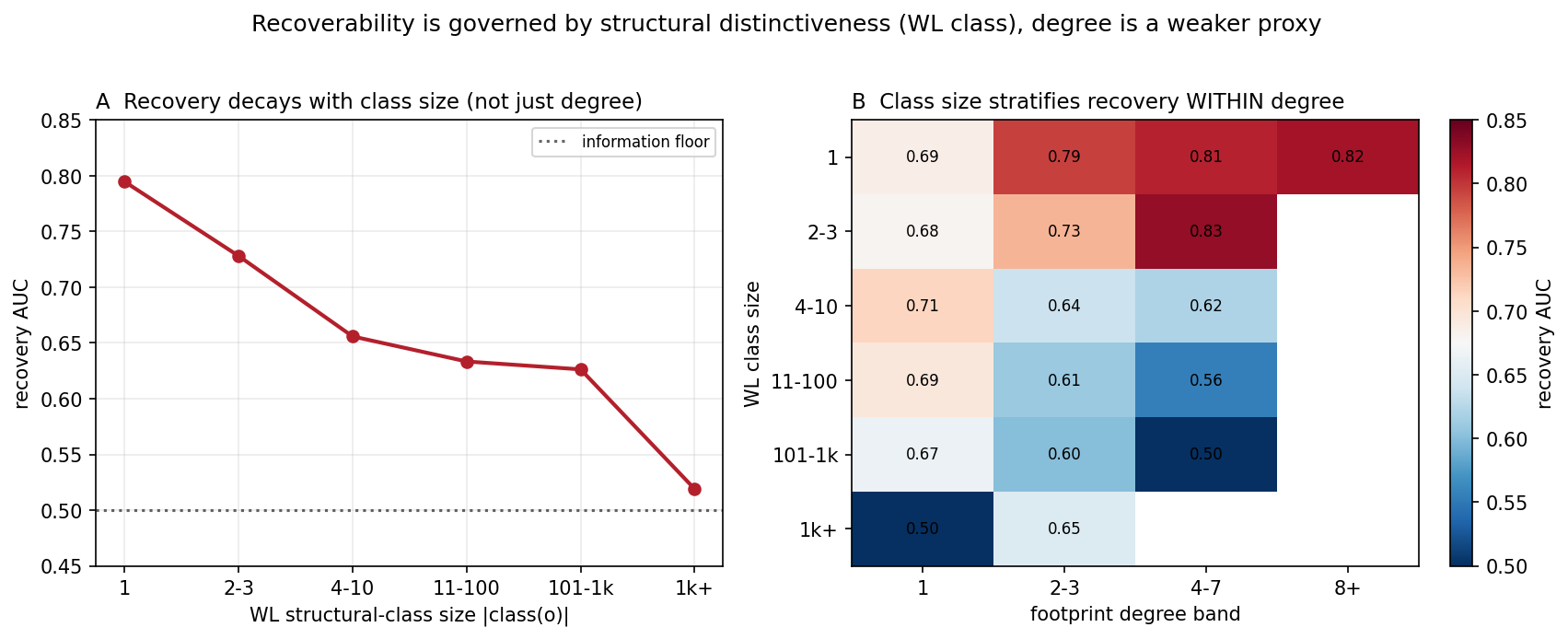}
  \caption{Left: recovery falls as WL colour-class size grows. Right: within each
  degree band, class size still separates recovery (colour is \textsc{auc}), so
  degree stands in for distinctiveness rather than causing it.}
  \label{fig:wl}
\end{figure*}

\subsection{Stronger expressivity does not lift the floor}\label{sec:seal}
Does a more expressive, subgraph-based predictor recover more, and does isolation
still hold it to $0.5$? We build a SEAL-style~\cite{seal} high-order model: a
gradient-boosted decoder over the enclosing-subgraph heuristics SEAL is meant to
capture, namely truncated Katz path counts, resource allocation, Adamic--Adar,
two-hop neighbourhood overlap, and DRNL-style distance indicators, all under the
same faithful pairing. Figure~\ref{fig:modelindep} places it beside the reach
probe and the GNN. At degree $0$ the high-order model is exactly $0.500$: with no
incident edge every subgraph heuristic is null and every distance label reads
``unreachable,'' so isolation is untouched by expressivity. Above the floor the
gain over the plain reach probe is small, at most about $0.02$, which again points
to reachable information as the binding constraint. This model is a heuristic
stand-in for a subgraph GNN. A full DGCNN-SEAL is left to future work, and it too
must sit on the degree-$0$ floor.

\section{Five Evaluation Traps}\label{sec:traps}
The floor cannot be beaten, and yet a handful of ordinary evaluation choices make
it look beaten. Each measures something other than the concealed endpoint's
identity. We met all five while validating the bound.

\textbf{Trap 1: same-address inverted separator.} Negatives that share the
entity's address turn an address feature into a separator pointing the wrong way,
since the hidden owner usually does \emph{not} share that address, and the
\textsc{auc} inflates.

\textbf{Trap 2: connected-decoy pool.} Negatives drawn from nodes connected to the
entity teach the classifier that the disconnected candidate is the owner. That is
circular, because the owner is defined by disconnection. With $76.2\%$ of owners
in a different component from the entity they own, this reports near-perfect
recovery that is really tautology.

\textbf{Trap 3: false-merge entity resolution.} Name-based resolution merges
common names almost every time, and the apparent identity gain is collision noise.
Stratifying by name distinctiveness removes it.

\textbf{Trap 4: entity-marginal leak.} At $d_{\mathrm{obs}}(o)=0$ a GNN decoder
score $d(h_o,h_e)$ reduces to $d(h_\emptyset,h_e)$, a function of the entity
alone. Varying the entity then scores the entity's propensity to carry a hidden
owner, which reads a spurious $0.765$ at degree $0$ (Figure~\ref{fig:gnn}).

\textbf{Trap 5: owner-marginal leak.} The symmetric error. Even when varying the
owner, drawing the decoy from all officers rather than from beneficial owners lets
the decoder reward ``looks like an owner'' over ``owns this entity,'' which lifts
degree-$1$ GNN recovery from a faithful $0.670$ to $0.876$
(Figure~\ref{fig:gnn}). Purely relational scores are immune, because their
features vanish with no $o$--$e$ path, which is why the reach probe needs no such
correction.

\begin{table}[t]
\caption{Spurious lift of each trap on beneficial-owner recovery. Traps 1, 2, 4, 5
report test \textsc{auc}. The honest baseline uses a faithful, same-population,
vary-owner evaluation.}
\label{tab:traps}
\small
\begin{tabular}{@{}clcc@{}}
\toprule
\# & \textbf{trap mechanism} & \textbf{honest} & \textbf{trap}\\
\midrule
1 & same-address inverted separator & 0.577 & \textbf{0.985}\\
2 & connected-decoy pool            & 0.577 & \textbf{0.989}\\
3 & false-merge entity resolution   & \multicolumn{2}{c}{$\sim\!25\%$ of credited resolution spurious$^\dagger$}\\
4 & entity-marginal leak            & 0.500 & \textbf{0.765}\\
5 & owner-marginal leak (deg.\ 1)   & 0.670 & \textbf{0.876}\\
\bottomrule
\end{tabular}
\\[2pt]{\footnotesize $^\dagger$For names shared by more than ten officers,
resolution fires almost every time, pure collision.}
\end{table}

Table~\ref{tab:traps} puts numbers on the inflation. Traps 4 and 5 are two faces
of one error, reading an endpoint's marginal in place of the pair, and both leave
the degree-$0$ floor untouched, since an isolated owner is a constant.

\smallskip\noindent\textbf{A faithful evaluation.} To test an identifiability
claim about endpoint $X$: remove all target edges from the graph used for
features; build the paired negative by varying $X$ with its partner fixed,
degree-matched or, better, orbit-matched, and drawn from the same population as
the true $X$; treat unlabeled pairs as positive-unlabelled rather than negative;
stratify any resolution step by name distinctiveness; and report results by the
endpoint's structural-class size. Varying the partner, or sampling the decoy from
a different population, measures a marginal instead of $X$.

\section{What \emph{Is} Recoverable}\label{sec:external}
\subsection{Attribute channels are capped}
The bound covers topology. Attributes can escape it in principle, and in practice
only barely (Table~\ref{tab:external}). Resolving a name to a connected same-name
node recovers edges the topology bound forbids, but for only about seven percent
of cases, and Trap~3 inflates even that. Joining an external record can hand an
isolated owner a footprint, and here coverage is the wall. Of isolated owners,
$94.6\%$ carry a name that appears nowhere else in the $814$K-node corpus.
Cross-source resolution lifts at most $0.7\%$ of targets in a way one can trust.
We count a same-name match as trustworthy only when the name is shared by at most
three officers across the whole corpus. Names carried by more than ten officers
resolve almost every time and are pure collision, which is Trap~3 at work, so we
exclude them from the trustworthy count. A real sanctions list (OFAC, $20$K names)
matches $0.04\%$ of owner names. The
one external source that carries the ownership edge outright, a public
beneficial-ownership register, covers exactly the jurisdictions that already
mandate disclosure, and not the offshore core (BVI, Cayman, Bermuda, Nevis, Samoa)
that this data is made of.

\begin{table}[t]
\caption{Coverage of the channels that could escape the topology bound. All are
capped at a few percent for beneficial owners.}
\label{tab:external}
\small
\begin{tabular}{@{}lr@{}}
\toprule
\textbf{Channel} & \textbf{coverage / lift}\\
\midrule
Isolated owners with a unique corpus-wide name & $94.6\%$\\
Internal cross-leak resolution (trustworthy)   & $\le 0.7\%$ of $R$\\
Name-resolution recovery channel               & $\sim 7\%$\\
OFAC sanctions name match                       & $0.04\%$\\
Public BO register coverage of offshore core    & structurally absent\\
\bottomrule
\end{tabular}
\end{table}

\subsection{The tractable target: the machinery of concealment}
One fact runs through every version of the problem. Any relation that ends at the
concealed party inherits the isolation limit, because ending there is what
concealment means. The nodes that stay visible are the ones doing the concealing.
This enabler layer is heavily concentrated: the top one percent of intermediaries
account for $62.7\%$ of all incorporation edges, a single mass-registration
address can link about $37{,}000$ nodes, and a few hundred intermediaries each
incorporate more than a thousand entities. Most individual nominees, by contrast,
appear only once or twice, so the signal lives in the hub tail. This layer is
tractable for the same reason the owners are not, namely that it is densely
connected, and it is the constructive target our negative result points toward.

\section{Threats to Validity}\label{sec:threats}
We group the main threats and, where we can, bound their direction.

\textbf{Label incompleteness.} Beneficial-owner labels are positive-unlabelled,
so a paired negative might in truth be an unextracted owner. The rate is tiny.
With $1.18$ beneficial-owner edges per owner over roughly $60$K target entities, a
random paired negative is a latent positive with probability below $2\times
10^{-5}$. Such noise can only pull mid-range \textsc{auc} down, which leaves our
recovery numbers conservative, and it cannot touch the degree-$0$ floor, a set of
exact ties.

\textbf{Orbits via a proxy.} Theorem~\ref{thm:floor} is exact for automorphism
orbits under $\mathrm{Stab}(e)$. Empirically we approximate these by global WL
colour classes, which are coarser and do not fix $e$. The approximation can only
overstate distinguishability, so the observed law, if anything, understates how
tightly recovery is bound to class size. The WL-matched-decoy probe of
Section~\ref{sec:results} gives direct evidence the proxy is tight enough to bite:
matching decoys on colour lowers recovery at every degree and holds the isolated
floor exactly. Validating class sizes against exact automorphism orbits on sampled
subgraphs would quantify the remaining slack, and we leave it to future work. The
isolated case needs no proxy and is exact.

\textbf{Choice of proximity signal.} The reach probe uses Adamic--Adar with a
discounted two-hop term. The SEAL-style predictor in Section~\ref{sec:seal}
substitutes a richer high-order feature set and moves the mid-range by at most
about $0.02$, and the endpoints---the exact $0.5$ at degree $0$ and the WL
law---do not move at all. We therefore expect the qualitative picture to survive
other reasonable proximity measures.

\textbf{Model family.} Our GNN evidence uses a SIGN/SGC encoder and a
gradient-boosted decoder, plus a high-order heuristic stand-in for a subgraph GNN.
A fully trained DGCNN-SEAL or a $k$-WL model would strengthen the mid-degree
story. None can escape the degree-$0$ floor, which follows from the theorem for
any isomorphism-invariant rule.

\textbf{Pruning of isolated nodes.} ICIJ removes isolated nodes before release,
which truncates the low end of the degree distribution. This makes our count of
isolated owners a lower bound, so the true share of unrecoverable owners is at
least the $20.6\%$ we report.

\textbf{Leak assignment.} We assign each target to a leak by its entity's source.
A handful of cross-source resolution edges blur these boundaries. The per-leak
floor is robust to this, since it depends only on the isolated owners within each
leak, which no assignment choice moves off the floor.

\textbf{External generalisation.} The numbers are specific to ICIJ. The theorem
is not, and it applies to any setting where a target relation is concealed by
severing an endpoint's connectivity. How far the empirical magnitudes carry to
other corpora is an open question.

\section{Discussion and Broader Impact}\label{sec:disc}
\textbf{Attributes and positional encodings.} Folding richer attributes into the
type map only shrinks $\mathrm{Aut}(G)$ and its orbits, so attributes can break
ties. That is the internal-attribute channel we already measure, worth about seven
percent, and not a topological escape. Positional encodings need care. Random
features and learnable per-node embeddings are not isomorphism-invariant, and for
an isolated owner they are independent of the target, so they add variance without
adding recoverable signal, and Corollary~\ref{cor:iso} stands. A transductive
node-ID embedding makes this concrete. An isolated test owner receives an
embedding untethered to any entity, since it has no edges through which training
could tie the two, so under the faithful evaluation it scores at chance. Any
apparent gain from such a model traces back to Traps~4 and~5 rather than to
recovery, which a same-population, vary-the-owner protocol removes. Laplacian
positional encodings of an isolated node are degenerate. Used in a
sign-invariant way they respect $\mathrm{Aut}(G)$ and the orbit argument goes
through, and used otherwise they become one more attribute channel under the same
coverage limits.

\textbf{Beyond 1-WL.} More expressive models, whether $2$-WL or subgraph GNNs
such as SEAL~\cite{seal} and BUDDY~\cite{buddy}, refine the classes and can raise
recovery on mid-degree strata where paths exist. At $d_{\mathrm{obs}}=0$ they
cannot help, since with no incident edge there is no path to carry information,
and Corollary~\ref{cor:iso} applies at any expressivity. Our SEAL-style predictor
bears this out in Figure~\ref{fig:modelindep}: exactly $0.500$ at degree $0$, and
only a slight mid-range gain. That predictor is a boosted model over
enclosing-subgraph heuristics rather than a trained subgraph GNN. We use it
because a full DGCNN-SEAL or BUDDY run is heavy at this scale, and scalable
variants such as S3GRL~\cite{s3grl} and Bloom-signature features~\cite{bloomsig}
are the right vehicles for it. A trained subgraph model is the clearest missing
baseline, and we expect it to lift the mid-range a little while landing on the
same degree-$0$ floor, which the theorem forces. Under a $k$-WL model the
distinctiveness law of Theorem~\ref{thm:floor}(ii) restates with $k$-WL classes,
and the sensitivity analysis, in which finer refinement strengthens the law, is
evidence the pattern persists.

\textbf{Policy.} The result is, in the end, a statement about the limits of
transparency-by-leak. The owners who matter most are the ones engineered to leave
no recoverable trace. Attribution at scale will come from mandatory public
beneficial-ownership registers, not from better inference over what has leaked. It
is also a caution against tools that claim to unmask hidden owners. As the five
traps show, such claims are easy to produce and, absent a faithful evaluation,
likely to be artefacts, with real consequences for the people they name.

\section{Conclusion}
For the isolated owners who dominate the data, recovering concealed ownership from
leak structure is a coin flip, and elsewhere it is governed by how distinctive the
owner is. A trained GNN and a high-order predictor both land on the same floor,
and the apparent exceptions turn out to be evaluation artefacts. We hope the bound
and the faithful-evaluation rule prove useful wherever link prediction meets
relations that were hidden on purpose.


\end{document}